\documentclass{ceurart}

\usepackage{listings}
\usepackage{centernot}
\usepackage{graphicx}
\usepackage{proof} % compact proof trees for examples
\usepackage{hyperref} % links
\usepackage{verbatim} % multiline comments
\usepackage[dvipsnames]{xcolor}
\usepackage{mathtools}
\usepackage{amsmath}
\usepackage{amssymb}
\usepackage{stmaryrd}
\usepackage{color}
\usepackage{amssymb}
\usepackage{pifont}
\usepackage{amsthm}
\usepackage{pifont}
\usepackage{stmaryrd}
\usepackage{listings}
\usepackage{color}
\usepackage{pgfplots}
\usepackage{subcaption}
\usepackage{cleveref}
\newif\ifsubmit
\submitfalse
\ifsubmit
 
 \newcommand{\DAComm}[1]{}
 
 \newcommand{\AFComm}[1]{}
\else
 
 \newcommand{\DAComm}[1]{{\scriptsize\textcolor{blue}{[\bf{Angelo: }#1}]}}
 
 \newcommand{\AFComm}[1]{{\scriptsize\textcolor{red}{[\bf{Angelo: }#1}]}}
\fi

\newcommand{\wrt}{w.~r.~t.}
\newcommand{\st}{s.~t.}

\newtheorem{definition}{Definition}[section]
\newtheorem{theorem}[definition]{Theorem}
\newtheorem{lemma}[definition]{Lemma}
\newtheorem{corollary}[definition]{Corollary}

\newcommand{\rn}[1]{\textsc{({#1})}} %% rule name
\newcommand{\Rule}[4]{{\tiny{\rn{#1}}}\displaystyle\frac{#2}{#3}\ \begin{array}{l} #4 \end{array}}

\newenvironment{grammar}{$\begin{array}[t]{llll}}{\end{array}$}
\newcommand{\production}[3]{#1&{:}{:}=&#2 & \mbox{{\small{#3}}}}

\newcommand{\term}[1]{\ensuremath{\mathtt {#1}}}
\newcommand{\nonterm}[1]{\ensuremath{\mathit {#1}}}
\newcommand{\re}{\nonterm{e}}
\newcommand{\reSet}{\nonterm{RE}}
\newcommand{\reSetExt}{\nonterm{RE+}}
\newcommand{\eps}{\term{\epsilon}}
\newcommand{\none}{\term{0}}
\newcommand{\sym}{\term{a}}
\newcommand{\asym}{\term{b}}
\newcommand{\yasym}{\term{c}}

\newcommand{\orop}{\mathbin{+}}
\newcommand{\shuffleop}{\mathbin{||}}
\newcommand{\catop}{}
\newcommand{\strcons}{\mathbin{:}}
\newcommand{\strcatop}{\mathbin{\cdot}}
\newcommand{\strshuffleop}{\shuffleop}

\newcommand{\starop}[1]{{{#1}\mathop{*}}}
\newcommand{\symAlph}{\Sigma}

\newcommand{\der}[1]{\xrightarrow{#1}}
\newcommand{\hasEps}[1]{\nu({#1})}
\newcommand{\conj}{\mathbin{\mathsf{and}}}
\newcommand{\disj}{\mathbin{\mathsf{or}}}

\newcommand{\parDer}[1]{\xrightharpoonup{#1}}

\newcommand{\incSym}{\Delta_{\max}}
\newcommand{\incSymSize}{\eta_{\max}}
\newcommand{\incFun}[1]{{\incSym}({#1})}
\newcommand{\incFunSize}[1]{{\incSymSize}({#1})}

\newcommand{\theight}[1]{{\mid}{#1}{\mid}} %% tree height
\newcommand{\tsizeSym}{{\mid\!\mid}}
\newcommand{\tsize}[1]{{\tsizeSym}{#1}{\tsizeSym}} %% tree height

\newcommand{\geqSym}{\mathit{geq}}
\newcommand{\geqFun}[2]{\geqSym({#1},{#2})}

\newcommand{\word}{\nonterm{w}}
\newcommand{\emptyWord}{\lambda}
\newcommand{\sem}[1]{\left\llbracket{#1}\right\rrbracket}
\newcommand{\opSem}[2]{\left\llbracket{#1}\right\rrbracket_{#2}}

\newcommand{\lang}{\nonterm{L}}

\begin{document}
%\copyrightyear{2025}
%\copyrightclause{Copyright for this paper by its authors. Use permitted under Creative Commons License Attribution 4.0 International (CC BY 4.0).}
\conference{Extended versions of the ICTCS 2025 paper with full proves}

\title{On The Space Complexity of Partial Derivatives of Regular Expressions with Shuffle}

\author[1]{Davide Ancona}[orcid=0000-0002-6297-2011,email=davide.ancona@unige.it]

\cormark[1]

\address[1]{University of Genova, Italy}

\author[2]{Angelo Ferrando}[orcid=0000-0002-8711-4670,
 email=angelo.ferrando@unimore.it
]

\address[2]{University of Modena and Reggio Emilia, Italy}

%% Footnotes
\cortext[1]{Corresponding author.}


\maketitle
\begin{abstract}
 Partial derivatives of regular expressions, introduced by Antimirov, define an elegant algorithm for generating equivalent non-deterministic finite automata (NFA) with a limited number of states.
 Here we focus on runtime verification (RV) of simple properties expressible with regular expressions. In this case, words are finite traces of monitorable events forming the language's alphabet, and the generated NFA may have an intractable number of states.
 This typically occurs when sub-traces of mutually independent events are allowed to interleave.
 To address this issue, regular expressions used for RV are extended with the shuffle operator to make specifications more compact and easier to read.
 Exploiting partial derivatives enables a rewriting-based approach to RV, where only one derivative is stored at each step, avoiding the construction of an intractably large automaton.
 This raises the question of the space complexity of the largest generated partial derivative. While the total number of generated partial derivatives is known to be linear in the size of the initial regular expression, no results can be found in the literature regarding the size of the largest partial derivative.
 We study this problem \wrt~ two metrics (height and size of regular expressions), and show that the former increases by at most one, while the latter is quadratic in the size of the regular expression. Surprisingly, these results also hold with shuffle.
\end{abstract}
\section{Introduction}\label{sec:intro}

Derivatives of regular expressions were introduced by Brzozowski \cite{Brzozowski64,OwensEtAl09}
to provide a very compact operational semantics for them, and to derive an algorithm for generating equivalent deterministic finite automata (DFA).
A first contribution of our work is to show that derivatives can be defined by means of a labeled transition system between regular expressions, where labels are the alphabet symbols. Given a fixed symbol, at each step there is always a unique transition, therefore the system defines a deterministic automaton where the initial state is an initial regular expression $\re$, and all other states are regular expressions that are derivatives obtained from $\re$.
Such an automaton is provably equivalent to $\re$,
if the set of its final states is defined as the set of all the corresponding regular expressions which recognize the empty word.
The main issue with this approach is that the number of syntactically different derivatives of $\re$ is generally unbounded, hence, to obtain an equivalent DFA, a finite quotient algebra of states modulo a suitable congruence has to be considered.

Inspired by Brzozowski's work, Antimirov \cite{Antimirov96} defined partial derivatives of regular expressions to define an algorithm that generates, from regular expressions, equivalent non-deterministic finite automata (NFA) with a limited number of states. The main breakthrough at the root of Antimirov's work is that the number of syntactically distinct partial derivatives of $\re$ is provably linear in the size of $\re$. Therefore, a construction analogous to that
proposed by Brzozowski can be used with syntactic equality without considering any weaker notion of congruence. However, the corresponding automaton is, in general non-deterministic, because
of the transition system defining partial derivatives: given a fixed symbol, at each step there can be zero or more transitions.

While Antimirov's work addresses the classical problem of word recognition in regular languages, here we focus on runtime verification (RV)~\cite{LeuckerSchallhart09} of simple properties expressible with regular expressions.

RV is a technique used to check the trace of events generated during a single execution of the System Under Scrutiny (SUS). This is done with monitors automatically generated from specifications that describe the correct behavior of the SUS (suitably instrumented to produce the trace to be analyzed).

RV is complementary to both formal verification and testing: like formal methods, it is based on specification languages, although it is not exhaustive; like testing, it is scalable and suitable for real-world systems and complex properties. Differently from testing, RV can detect errors even in non-deterministic systems~\cite{HavelundRosu2004,SharmaEtAl2009,TowardsIoT17,SchiavioEtAl19}, and therefore represents a valuable complement to testing approaches~\cite{BesnardEtAl23}.

When properties to be verified can be defined with regular expressions, words are finite traces of monitorable events, which define the alphabet of the language, and the monitor is the finite automaton generated from the regular expression defining the correct traces. However, this approach may be unfeasible because of the intractable number of states of the generated automaton, even for NFAs.

This typically happens when the properties to be verified consist of sub-traces of mutually independent events, which are allowed to interleave during the execution of the SUS. Consider, for instance, the challenge of verifying that files are properly opened and closed, in a system able to handle multiple files independently. In such cases, extending regular expressions with the shuffle operator \cite{RML2021}
can significantly reduce the size of specifications and enhance clarity.
Moreover, by using partial derivatives, it becomes possible to avoid generating automata with an intractable number of states, and adopt, instead, a rewriting-based approach~\cite{RosuEtAl2005,RML2021}  to RV, where only a single partial derivative needs to be considered at each reduction step.

This raises the question of the space complexity of the largest partial derivative that can be generated from an initial regular expression. Indeed, while the total number of generated partial derivatives is known to be linear in the size of the initial regular expression, no results can be found in the literature regarding the size of the largest partial derivative.

In this paper, we investigate this problem with respect to two different metrics for regular expressions: their height and their total number of nodes when viewed as trees. In particular, we show that the height of the largest partial derivative can increase by at most one, while the number of nodes in the largest partial derivative is bounded by $n^2$, where $n$ is the number of nodes of the initial regular expression. Surprisingly, these results still hold when regular expressions are extended with shuffle.

To this aim, we propose a proof methodology based on the definition of a function that, given a regular expression, returns an upper bound of the increment w.r.t. a specific metric for all its  partial derivatives. In this way, an invariant can be established on single rewriting steps and directly extended to multiple steps, to derive the expected space complexity results.  This proof methodology allows us to keep the same proof structure for all main results, and reuse parts of the proofs.

The paper is structured as follows: \Cref{sec:der} introduces the basic notion of derivatives of regular expressions. \Cref{sec:parDer} defines partial derivatives and proves the space complexity results ~\wrt~the height and the size of expressions.
\Cref{sec:shuffle} extends regular expressions with the shuffle operator, and show that the property proved in \cref{sec:parDer} still hold. Finally, \cref{sec:conclu} draws conclusions and some directions for future work.
\section{Derivatives of regular expressions}\label{sec:der}

This section provides the basic notions on derivatives, by assuming familiarity with regular expressions.

Let $\symAlph$ denote a given non-empty finite set of symbols, called \emph{alphabet}. We call \emph{word} an element of $\symAlph^*$, and denote with $\emptyWord$ the empty word. Furthermore, $\sym\strcons\word$ denotes the word where $\sym$ is the first symbol, and $\word$ the rest of the word, while $\word\strcatop\word'$ denotes word concatenation.
The set $\reSet$ of regular expressions over $\symAlph$ is inductively defined by the following grammar:
\begin{center}
 \begin{small}
  \begin{grammar}
   \production{\re}{\none \mid \eps \mid \sym \mid \re_0\catop\re_1 \mid \re_0\orop\re_1 \mid \starop{\re}}{with $\sym\in\symAlph$} \\
  \end{grammar}
 \end{small}
\end{center}
We use the standard precedence rules: the Kleene star operator $\starop{\re}$  has higher precedence than concatenation $\re_0\catop\re_1$, which in turn has higher precedence than union $\re_0\orop\re_1$.

\Cref{def:langOp} introduces the standard operators on regular languages, and \cref{def:sem} the standard semantics of regular expressions.

\begin{definition}\label[definition]{def:langOp}
 For all $\lang,\lang_0,\lang_1\subseteq\symAlph^*$
 \[
  \begin{array}{l}
   \lang_1\strcatop\lang_2=\{\word_1\strcatop\word_2\mid \word_1\in\lang_1,\word_2\in\lang_2\} \quad
   \lang^0=\{\emptyWord\}, \lang^{n+1}=\lang\strcatop\lang^n \mbox{for all $n\geq 0$}           \quad
   \starop{\lang} = \bigcup_{n\geq 0}\lang^n
  \end{array}
 \]
\end{definition}
\begin{definition}\label[definition]{def:sem}
 \[
  \begin{array}{l}
   \sem{\none}=\emptyset\quad \sem{\eps}=\{\emptyWord\} \quad \sem{\sym}=\{\sym\} \quad
   \sem{\re_0\catop\re_1}=\sem{\re_0}\strcatop\sem{\re_1} \quad \sem{\re_0\orop\re_1}=\sem{\re_0}\cup\sem{\re_1} \quad \sem{\starop{\re_0}}=\starop{\sem{\re_0}}
  \end{array}
 \]
\end{definition}

Note the difference between $\eps$ (a constant in the syntax of regular expressions) and $\emptyWord$ (the empty word).

The derivative  \cite{Brzozowski64} of a language $L$ \wrt~a word $\word\in\symAlph^*$, denoted by $\word^{-1}(\lang)$, is defined as follows:
\begin{definition}\label[definition]{def:sem-der}
 For all $\lang\subseteq\symAlph^*$, $\word\in\symAlph^*$,\quad
 $\word^{-1}(\lang) = \{ \word' \mid \word\strcatop\word'\in\lang \}$.

\end{definition}
Directly from \cref{def:sem-der} one can derive
\begin{equation}
 \label{eq:der}
 \word\in\lang \mbox{ iff }
 \emptyWord\in\word^{-1}(\lang)
\end{equation}

While the notion of derivative applies to any formal language, regular languages enjoy the following closure property: if $\lang$ is regular, then
$\word^{-1}(\lang)$ is regular as well, for any word $\word$. In other words, the derivative of a regular expression can be defined by another regular expression.

This property allows an operational and succinct definition for the derivatives of regular expressions.
First, the derivative \wrt~a single symbol is defined (see~\cref{fig:der}), then
the more general notion of derivative \wrt~a word is derived by reflexive  transitive closure (see~\cref{def:der}).

Here we deliberately follow the style of a labelled transition system, defined by means of inference rules, to highlight two interesting related aspects:\footnote{To the best of our knowledge, no previous work has used this style to define the (partial) derivatives of regular expressions.}
\begin{itemize}
 \item derivatives provide an intuitive way to define a small-step semantics for regular expressions;
 \item a reduction step $\re\der{\sym}\re'$ corresponds to a transition step from the state represented by $\re$ to a state represented by $\re'$ with symbol $\sym$, for an automaton which recognizes the language defined by $\re$.
\end{itemize}
By using the notation of Antimirov~\cite{Antimirov96}, which is an adaptation of that introduced by Brzozowski~\cite{Brzozowski64}, we have
that $\sym^{-1}(\re_0)=\re_1$ iff $\re_0\der{\sym}\re_1$, and $\word^{-1}(\re_0)=\re_1$ iff $\re_0\der{\word}\re_1$.

The definition of $\hasEps{\re}$ in~\cref{fig:der} deserves some comments:
the intended meaning is that $\hasEps{\re}=\eps$ iff
$\emptyWord\in\sem{\re}$, and, dually, $\hasEps{\re}=\none$ iff
$\emptyWord\not\in\sem{\re}$. The base cases of the definition are straightforward; thanks to the auxiliary functions $\conj$, $\disj$ defined on $\eps$ and $\none$ in the corresponding tables, we can deduce that, as expected, $\re_0\catop\re_1$ recognizes $\emptyWord$ iff both $\re_0$ and $\re_1$ recognize $\emptyWord$, while  $\re_0\orop\re_1$ recognizes $\emptyWord$ iff $\re_0$ or $\re_1$ recognizes $\emptyWord$.

The fact that $\hasEps{\re}$ does not return a Boolean value was a deliberate choice of Brzozowski to define rule \rn{cat} in a more compact way:
if $\hasEps{\re_0}=\eps$, then $\re_0\catop\re_1\der{\sym}\re'_0\orop\eps\catop\re_1'$, since in this case $\re_1$ contributes to the
derivative. Otherwise, $\re_0\catop\re_1\der{\sym}\re'_0\orop\none\catop\re_1'$, since $\re_1$ does not contribute to the
derivative.

Note that by definition of the rules, for all $\re$ and $\sym$, there always exists a unique derivative of $\re$ \wrt~$\sym$.
For instance, if $\sym\neq\asym$, $\sym\neq\yasym$, then
$\sym \asym \orop \sym\yasym\der{\sym}(\eps\catop\asym\orop\none\catop\none)\orop(\eps\catop\yasym\orop\none\catop\none)$ and $\sym \asym \orop \sym \yasym\der{\asym}(\none\catop\asym\orop\none\catop\eps)\orop(\none\catop\yasym\orop\none\catop\none)$.
\begin{figure}[t]
 \begin{center}
  $$
   \begin{array}{c}
    \Rule{empty}{}{\none\der{\sym}\none}\quad
    \Rule{eps}{}{\eps\der{\sym}\none}\quad
    \Rule{sym-eq}{}{\sym\der{\sym}\eps}\quad
    \Rule{sym-neq}{}{\sym\der{\asym}\none}{\sym\neq\asym} \\[4ex]
    \Rule{cat}{\re_0\der{\sym}\re'_0\quad\re_1\der{\sym}\re'_1}{\re_0\catop\re_1\der{\sym}\re'_0\catop\re_1\orop\hasEps{\re_0}\catop\re'_1}{} \quad
    \Rule{or}{\re_0\der{\sym}\re'_0\quad\re_1\der{\sym}\re'_1}{\re_0\orop\re_1\der{\sym}\re'_0\orop\re'_1}{}\quad
    \Rule{star}{\re\der{\sym}\re'}{\starop{\re}\der{\sym}\re'\catop\starop{\re}}{}
   \end{array}
  $$
  \begin{tabular}{cccc}
   $
    \begin{array}{l}
     \hasEps{\eps}=\hasEps{\starop{\re_0}}=\eps                  \\
     \hasEps{\none}=\hasEps{\sym}=\none                          \\
     \hasEps{\re_0\catop\re_1}=\hasEps{\re_0}\conj\hasEps{\re_1} \\
     \hasEps{\re_0\orop\re_1}=\hasEps{\re_0}\disj\hasEps{\re_1}
    \end{array}
   $
    &
   \begin{tabular}{|c|c|c|}
    \hline
    $\conj$ & $\quad\none\quad$ & $\quad\eps\quad$ \\
    \hline
    $\none$ & $\none$           & $\none$          \\
    \hline
    $\eps$  & $\none$           & $\eps$           \\
    \hline
   \end{tabular}
    & \qquad &
   \begin{tabular}{|c|c|c|}
    \hline
    $\disj$ & $\quad\none\quad$ & $\quad\eps\quad$ \\
    \hline
    $\none$ & $\none$           & $\eps$           \\
    \hline
    $\eps$  & $\eps$            & $\eps$           \\
    \hline
   \end{tabular}
  \end{tabular}
 \end{center}

 \caption{Transition system defining the derivatives of regular expressions}
 \label{fig:der}
\end{figure}

\Cref{def:der} introduces the general notion of derivative \wrt~a word $\word$, by computing the reflexive transitive closure of the one-step relation defined in~\cref{fig:der}. As expected, the definition is by induction\footnote{It is worth noting that the computation of a derivative always ``terminates'' in a finite number of steps, because words have finite length.} on the length of $\word$. The multiple-steps rewriting relation provides also the operational semantics of regular expressions \wrt~derivatives, driven by~\cref{eq:der}. In the view above, which considers the labelled transition system as a deterministic automaton,
the definition corresponds to that of accepted language, where the set of final states is the set of all derivatives $\re'$ \st~$\hasEps{\re'}=\eps$.
It is straightforward to prove that the set of all syntactically distinct derivatives of a regular expression $\re$ is generally unbounded, and so is their size. Hence, the underlying automaton is infinite, unless some finite quotient algebra of states modulo a suitable congruence is considered.
\begin{definition}\label[definition]{def:der}
 For all $\re,\re_0,\re_1\in\reSet$, $\sym\in\symAlph$, $\word\in\symAlph^*$
 \[
  \begin{array}{l}
   \re\der{\emptyWord}\re \qquad
   \re_0\der{\sym\strcons\word}\re_1 \mbox{ iff } \re_0\der{\sym}\re \mbox{ and } \re\der{\word}\re_1
   \\
   \opSem{\re}{\der{}}=\{\word\in\symAlph^* \mid \mbox{there exists } \re'\in\reSet \mbox{ s.t. } \re\der{\word}\re' \mbox{ and } \hasEps{\re'}=\eps\}
  \end{array}
 \]
\end{definition}

\Cref{theo:der} establishes the results proved by Brzozowski \cite{Brzozowski64}:
the definitions of $\hasEps{\re}$ and $\re\der{\word}\re'$ are sound
\wrt\ their intended semantics, and the operational semantics of regular expressions is equivalent to the standard one.
\begin{theorem}\label[theorem]{theo:der}
 For all $\re,\re'\in\reSet$, $\word\in\symAlph^*$
 \[
  \begin{array}{l}
   \hasEps{\re} = \eps \mbox{ iff } \emptyWord\in\sem{\re}              \qquad
   \re\der{\word}\re' \mbox{ implies } \word^{-1}(\sem{\re})=\sem{\re'} \qquad
   \opSem{\re}{\der{}}=\sem{\re}
  \end{array}
 \]
\end{theorem}

\section{Partial derivatives of regular expressions}\label{sec:parDer}

\begin{figure}
 $$
  \begin{array}{c}
   \Rule{sym}{}{\sym\parDer{\sym}\eps}\quad
   \Rule{l-cat}{\re_0\parDer{\sym}\re'_0}{\re_0\catop\re_1\parDer{\sym}\re'_0\catop\re_1}{}\quad
   \Rule{r-cat}{\re_1\parDer{\sym}\re'_1}{\re_0\catop\re_1\parDer{\sym}\re'_1}{\hasEps{\re_0}=\eps} \\[4ex]
   \Rule{l-or}{\re_0\parDer{\sym}\re'_0}{\re_0\orop\re_1\parDer{\sym}\re'_0}{} \quad
   \Rule{r-or}{\re_1\parDer{\sym}\re'_1}{\re_0\orop\re_1\parDer{\sym}\re'_1}{} \quad
   \Rule{star}{\re\parDer{\sym}\re'}{\starop{\re}\parDer{\sym}\re'\catop\starop{\re}}{}
  \end{array}
 $$
 \caption{Transition system defining the partial derivatives of regular expressions}
 \label{fig:parDer}
\end{figure}

~\Cref{fig:parDer} contains the transition rules defining the partial derivatives of regular expressions \wrt~ a single symbol.
The definition follows the work of Antimirov \cite{Antimirov96}, but
as happens for the definition of derivatives, our presentation is based on a labelled transition system. More importantly, here we deliberately define a non-deterministic system, where a single transition step yields a single partial derivative among all possible ones, while in Antimirov's definition the whole set of partial derivatives is computed.\footnote{This is for efficiency reasons, though the obtained states are still used by Antimirov to generate an NFA.} With our approach, the labelled transition system directly induces an NFA equivalent to the initial regular expression, and the main differences between derivatives and partial derivatives are highlighted, as explained below.

Rules \rn{cat} and \rn{or} in \cref{fig:der} are split in \cref{fig:parDer} in the two possibly overlapping\footnote{In the sense that for some regular expressions and symbols, both are applicable.} rules \rn{l-cat}, \rn{r-cat} and \rn{l-or}, \rn{r-or}, respectively.
For instance, according to Antimirov's definition\footnote{See \cite{Antimirov96}, definition 2.8.},
the partial derivative of $\sym \asym \orop \sym\yasym$ \wrt~symbol $\sym $ returns the set of regular expressions $\{\eps\catop\asym,\eps\catop\yasym\}$ (that is, $\delta_\sym(\sym\asym \orop \sym\yasym)=\{\eps\catop\asym,\eps\catop\yasym\}$ with Antimirov's notation). This corresponds to the fact that there exist two possible transition steps from $\sym \asym \orop \sym\yasym$ labeled with $\sym $, namely,
$\sym \asym \orop \sym\yasym\parDer{\sym}\eps\catop\asym$ and $\sym \asym \orop \sym \yasym\parDer{\sym}\eps\catop\yasym$.

Another difference with derivatives is that in \cref{fig:parDer} there are no rules returning the empty derivative. Therefore, there are cases where the partial derivative of a regular expression \wrt\ a symbol (or word) does not exist. This is useful when partial derivatives are used for RV, because in this way errors can be detected with no further checks on the partial derivative: if no moves can be taken, then the partial derivative is empty and  the trace of events (that is, the word) considered so far cannot be the prefix of any correct trace, therefore a violation of the specification can be reported.

% While Antimirov's definition is still deterministic since sets  are considered instead of single partial derivatives, here we deliberately introduce a non-deterministic system.
% This difference is motivated by the main aim of our work: Antimirov's approach uses partial derivatives for generating an NFA, while here
% we use them to define a non-deterministic transition system which keeps the underlying NFA implicit.

Finally, a major difference between derivatives and partial derivatives
lies in the fact that the set of all syntactically distinct partial derivatives of a given regular expression $\re$ is bounded by the size of $\re$. This means that, after a single partial derivative is generated at each rewriting step, there is no need to simplify it in order to keep the size of partial derivatives bounded.

\Cref{def:parDer} is dual to \Cref{def:der}.
\begin{definition}\label{def:parDer}
 For all $\re,\re_0,\re_1\in\reSet$, $\sym\in\symAlph$, $\word\in\symAlph^*$
 \[
  \begin{array}{l}
   \\
   \re\parDer{\emptyWord}\re                                                                                       \qquad
   \re_0\parDer{\sym\strcons\word}\re_1 \mbox{ iff } \re_0\parDer{\sym}{}\re \mbox{ and } \re\parDer{\word}{}\re_1 \\
   \opSem{\re}{\parDer{}}=\{\word\in\symAlph^* \mid \mbox{there exists } \re'\in\reSet \mbox{ s.t. } \re\parDer{\word}\re' \mbox{ and } \hasEps{\re'}=\eps\}
  \end{array}
 \]
\end{definition}

The claims below are dual to \cref{theo:der}, and  directly follow from the results of Antimirov \cite{Antimirov96}.

Let $\delta_\word(\re)$ denotes the set of all partial derivatives of $\re\in\reSet$ \wrt~ $\word\in\symAlph^*$:
\[
 \delta_\word(\re)=\{\re'\in\reSet \mid \re\parDer{\word}\re'\}.
\]

% \begin{definition}
%  Let $\reSet_1$ and $\reSet_2$ be sets of regular expressions.
%  \[
%   \begin{array}{l}
%    \reSet_1\derSet{\sym}\reSet_2 \mbox{ iff } \reSet_2=\{\re_1 \mid \mbox{ there exists } \re_0\in\reSet_1 \mbox{ s.t. } \re_0\parDer{\sym}{}\re_1 \} \\
%    \hasEps{\reSet}=\eps \mbox{ iff there exists } \re\in\reSet \mbox{ s.t. } \hasEps{\re}                                                             \\
%    \re_0\derSet{\sym\word}\re_1 \mbox{ iff } \re_0\derSet{\sym}\re \mbox{ and } \re\derSet{\word}\re_1                                                \\
%    \re\derSet{\emptyWord}{}\re                                                                                                                        \\
%    \opSem{\re}{\derSet{}}=\{\word\in\wordUniv \mid \mbox{there exists } \reSet \mbox{ s.t. } \{\re\}\derSet{\word}\reSet \mbox{ and } \hasEps{\reSet}=\eps\}
%   \end{array}
%  \]
% \end{definition}

\begin{theorem}\label[theorem]{theo:parDer}
 For all $\re\in\reSet$,
 $\word^{-1}(\sem{\re})=\bigcup_{\re'\in\delta_\word(\re)}\sem{\re'}$ and
 $\opSem{\re}{\parDer{}{}}=\opSem{\re}{\der{}}$.
\end{theorem}

\begin{corollary}\label[corollary]{cor:parDer}
 For all $\re\in\reSet$,
 $\opSem{\re}{\parDer{}{}}=\sem{\re}$.
\end{corollary}
% \begin{theorem}
%  For all regular expressions $\re$
%  \[
%   \opSem{\re}{\derSet{}}=\opSem{\re}{\der{}}
%  \]
% \end{theorem}

The following theorem is based as well on the results proved by Antimirov \cite{Antimirov96}.

\begin{theorem}\label[theorem]{theo:boundParDer}
 For all $\re\in\reSet$, $\delta_\word(\re)$ is bounded.
\end{theorem}

\subsection{Height of partial derivatives}\label{sec:height}

In this section we investigate the upper bound on the height of the partial derivatives of a given regular expression, defined in a standard way as follows:
$$
 \begin{array}{c}
  \theight{\eps}=\theight{\sym}=0                                                              \qquad
  \theight{\re_0\catop\re_1}=\theight{\re_0\orop\re_1}=\max(\theight{\re_0},\theight{\re_1})+1 \qquad
  \theight{\starop{\re}}=\theight{\re}+1
 \end{array}
$$
Since the height of the partial derivative of an expression $\re$  \wrt~ a word $\word$ both depends on the shape of $\re$ and on $\word$, we need to reason on any possible $\re$ and $\word$. Moreover, we aim at providing a general proof methodology which can be followed to prove results when considering also shuffle and the size of expressions.

To show a couple of examples let us consider the following reduction steps which compute the partial derivatives of $\starop{\sym}\catop\starop{\asym}$~
\wrt~ $\sym\asym$ and $\asym\asym$:
\begin{flushleft}
 $
  \begin{array}{l}
   \starop{\sym}\catop\starop{\asym}\parDer{\sym}(\eps\catop\starop{\sym})\catop\starop{\asym}\parDer{\asym}\eps\catop\starop{\asym} \qquad
   \starop{\sym}\catop\starop{\asym}\parDer{\asym}\eps\catop\starop{\asym}\parDer{\asym}\eps\catop\starop{\asym} \\
   \theight{\starop{\sym}\catop\starop{\asym}}=2\qquad \theight{(\eps\catop\starop{\sym})\catop\starop{\asym}}=3 \qquad \theight{\eps\catop\starop{\asym}}=2
  \end{array}
 $
\end{flushleft}
In the first example the height increases by one after the first step and decreases by one after the second, while in the second example the height of the expressions is unchanged.

What we are going to prove are the following main properties:
\begin{enumerate}
 \item the height of the partial derivative of a regular expression \wrt~ a symbol can increase \textbf{at most by one};
 \item the height of the partial derivative of a partial derivative \wrt~a non-empty word, \textbf{cannot} increase;
 \item from 1. and 2. and from the definition of partial derivative, one can deduce that if $\re\parDer{\word}\re'$, then $\theight{\re'}\leq\theight{\re}+1$.
\end{enumerate}
Although these properties can be proved in a direct way, they must be adapted when the size of expressions is considered.
Furthermore, they do not hold in this stronger form, when shuffle is added. Therefore we provide a more general proof methodology which works for both metrics on expressions, and can be adapted to the case of shuffle, where weaker claims hold.

In order to do that, we define the function $\incFun{\re}$ which returns an upper bound for $\theight{\re'}-\theight{\re}$, where
$\re\parDer{\word}\re'$, with $\re'\in\reSet$, $\word\in\symAlph^*$.
That is, if $M_\re\stackrel{\mathrm{def.}}{=}\max\{\theight{\re'}-\theight{\re} \mid \re\parDer{\word}\re', \re'\in\reSet, \word\in\symAlph^*\}$, then
a sound definition of $\incFun{\re}$ has to satisfy the constraint $\incFun{\re}\geq M_\re$. Hence, $\theight{\re}+\incFun{\re}$ provides an upper bound for the height of all partial derivatives of $\re$: $\theight{\re}+\incFun{\re}\geq\max\{\theight{\re'} \mid \re\parDer{\word}\re', \re'\in\reSet, \word\in\symAlph^*\}$.
\begin{definition}\label[definition]{fig:incFun}
 The function $\incFun{\re}$ is recursively defined as follows:
 $$
  \begin{array}{ll}
   \incFun{\sym}=\incFun{\eps}=0                                      \qquad
                              & \incFun{\re_0\catop\re_1}=\geqFun{\re_0}{\re_1}\cdot\incFun{\re_0} \\
   \incFun{\re_0\orop\re_1}=0 & \incFun{\starop{\re}}=1
  \end{array}
 $$
 where
 $
  \geqFun{\re_0}{\re_1}     =
  \begin{cases}
   1, & \text{if } \theight{\re_0} \geq \theight{\re_1} \\
   0, & \text{if } \theight{\re_0} < \theight{\re_1}
  \end{cases}
 $
\end{definition}
Before commenting the definition of $\incFun{\re}$, we introduce the following straightforward lemma.
\begin{lemma}\label[lemma]{lemma:bounds}
 For all $\re\in\reSet$,
 $0\leq\incFun{\re}\leq 1$.
\end{lemma}
\begin{proof}
 Directly by induction on the definition of $\incFun{\re}$.
\end{proof}

Because $\re\parDer{\emptyWord}\re$, $0\in \{\theight{\re'}-\theight{\re} \mid \re\parDer{\word}\re', \re'\in\reSet, \word\in\symAlph^*\}$ and the constraint $0\leq\incFun{\re}$ must always be met. The constraint $\incFun{\re}\leq 1$ corresponds to what we want to prove in item 3 above.

The definition of $\incFun{\re}$ is driven by the reduction rules in \cref{fig:parDer} and by the property in item 3. We comment only the non-trivial cases: if $\re=\re_0\orop\re_1$, then we expect $\theight{\re'_i}\leq\theight{\re_i}+1$, $i=0,1$,
hence $\theight{\re'_i}\leq\theight{\re}$ and $\incFun{\re}=0$;
if $\re=\starop{\re_0}$, then we expect $\theight{\re'_0}\leq\theight{\re_0}+1$,
hence $\theight{\re'_0\catop\starop{\re_0}}\leq\theight{\re}+1$ and $\incFun{\re}=1$. Finally, the definition for $\re=\re_0\catop\re_1$ requires more care:
if rule $\rn{l-cat}$ is applied, then $\theight{\re_0'\catop\re_1}=\theight{\re}+1$, but only when $\theight{\re_0'}=\theight{\re_0}+1$ and $\theight{\re_0}\geq\theight{\re_1}$, otherwise $\theight{\re_0'\catop\re_1}\leq\theight{\re}$; if rule $\rn{r-cat}$ is applied, then $\theight{\re'_1}\leq\theight{\re}$, as happens for rule $\rn{r-or}$. Therefore $\incFun{\re}=1$ iff $\incFun{\re_0}=1$ and $\theight{\re_0}\geq\theight{\re_1}$, that is, $\geqFun{\re_0}{\re_1}=1$. In all other cases, that is $\incFun{\re_0}=0$ or $\geqFun{\re_0}{\re_1}=0$, $\incFun{\re}=0$.

The following theorem proves the soundness of the definition of $\incSym$ \wrt~ its intended meaning. As a consequence, the height of the derivative of an expression $\re$ \wrt~ a symbol is always bounded by $\theight{\re}+1$.
\begin{theorem}\label[theorem]{theo:inc-bound}
 For all $\re,\re'\in\reSet$, $\sym\in\symAlph$, if $\re\parDer{\sym}\re'$, then $\theight{\re'}\leq \theight{\re}+\incFun{\re}$.
\end{theorem}
\begin{proof}
 By induction and case analysis on the rules defining $\re\parDer{\sym}\re'$.
 \begin{description}
  \item[base case:] The only base rule is \rn{sym}.
   Hence, $\re=\sym$, $\re'=\eps$, and
   $\theight{\eps}=0=\theight{\sym}=\theight{\sym}+0=\theight{\sym}+\incFun{\sym}$.
  \item[inductive step:] the proof proceeds by case analysis on the rule applied
   at the root of the derivation tree.
   \begin{itemize}
    \item $\Rule{l-cat}{\re_0\parDer{\sym}\re'_0}{\re_0\catop\re_1\parDer{\sym}\re'_0\catop\re_1}{}$\\[2ex]
          By inductive hypothesis $\theight{\re_0'}\leq \theight{\re_0}+\incFun{\re_0}$.
          We distinguish two cases:
          \begin{itemize}
           \item $\theight{\re_0}\geq\theight{\re_1}$:
                 $\theight{\re_0'\catop\re_1}\stackrel{\mathrm{def}}{=}\max(\theight{\re'_0},\theight{\re_1})+1\leq\max(\theight{\re_0}+\incFun{\re_0},\theight{\re_1})+1\leq\max(\theight{\re_0},\theight{\re_1})+1+\incFun{\re_0}=\theight{\re_0\catop\re_1}+\incFun{\re_0}=\theight{\re_0\catop\re_1}+\incFun{\re_0\catop\re_1}$, by the inductive hypothesis, the definition of $\max$, $\theight{\ }$, $\incSym$, and $\geqSym$, the assumption $\theight{\re_0}\geq\theight{\re_1}$, and \cref{lemma:bounds}.

           \item $\theight{\re_0}<\theight{\re_1}$:
                 $\theight{\re_0'\catop\re_1}\stackrel{\mathrm{def}}{=}\max(\theight{\re'_0},\theight{\re_1})+1\leq\max(\theight{\re_0}+\incFun{\re_0},\theight{\re_1})+1=\theight{\re_0\catop\re_1}=\theight{\re_0\catop\re_1}+\incFun{\re_0\catop\re_1}$, by the inductive hypothesis, the definition of $\max$, $\theight{\ }$, $\incSym$, and $\geqSym$, the assumption $\theight{\re_0}<\theight{\re_1}$, and \cref{lemma:bounds}.
          \end{itemize}

    \item $\Rule{r-cat}{\re_1\parDer{\sym}\re'_1}{\re_0\catop\re_1\parDer{\sym}\re'_1}{\hasEps{\re_0}=\eps}$\\[2ex]
          By inductive hypothesis $\theight{\re_1'}\leq \theight{\re_1}+\incFun{\re_1}$.

          Therefore, $\theight{\re_1'}\leq \theight{\re_1}+\incFun{\re_1}\leq\theight{\re_1}+1\leq\theight{\re_0\catop\re_1}\leq\theight{\re_0\catop\re_1}+\incFun{\re_0\catop\re_1}$, by the inductive hypothesis, the definition of $\max$, $\theight{\ },$ and $\incSym$, and \cref{lemma:bounds}.

    \item $\Rule{l-or}{\re_0\parDer{\sym}\re'_0}{\re_0\orop\re_1\parDer{\sym}\re'_0}{}$\\[2ex]
          By inductive hypothesis $\theight{\re_0'}\leq \theight{\re_0}+\incFun{\re_0}$.

          Therefore, $\theight{\re_0'}\leq \theight{\re_0}+\incFun{\re_0}\leq\theight{\re_0}+1\leq\theight{\re_0\orop\re_1}=\theight{\re_0\orop\re_1}+\incFun{\re_0\orop\re_1}$, by the inductive hypothesis, the definition of $\max$, $\theight{\ },$ and $\incSym$, and \cref{lemma:bounds}.
    \item rule \rn{r-or} is symmetric to  rule \rn{l-or}.
    \item $\Rule{star}{\re\parDer{\sym}\re'}{\starop{\re}\parDer{\sym}\re'\catop\starop{\re}}{}$\\[2ex]
          By inductive hypothesis $\theight{\re'}\leq \theight{\re}+\incFun{\re}$.

          Therefore, $\theight{\re'\catop\starop{\re}}\stackrel{\mathrm{def}}{=}\max(\theight{\re'},{\theight{\starop{\re}}})+1\leq \max(\theight{\re}+\incFun{\re},{\theight{\re}}+1)+1=\theight{\re}+1+1=\theight{\starop{\re}}+1=\theight{\starop{\re}}+\incFun{\starop{\re}}$, by the inductive hypothesis, the definition of $\max$, $\theight{\ },$ and $\incSym$, and \cref{lemma:bounds}.
   \end{itemize}
 \end{description}
\end{proof}

The following corollary can be directly derived from \cref{theo:inc-bound} and \cref{lemma:bounds}.
\begin{corollary}\label[corollary]{cor:bound}
 For all $\re,\re'\in\reSet$, and $\sym\in\symAlph$, if $\re\parDer{\sym}\re'$, then $\theight{\re'}\leq\theight{\re}+1$.
\end{corollary}
~\Cref{theo:inc-bound} shows that the height of the derivative of an expression $\re$ \wrt~ a symbol is bounded by $\theight{\re}+1$.
The following results prove a strong property: the height of the partial derivative of a partial derivative cannot increase; that is, after the first reduction step, the height of the derivative of an expression $\re$ \wrt~ a symbol is bounded by $\theight{\re}$.
Therefore, one can conclude that the height of the derivative of an expression $\re$ \wrt~ an arbitrary  word (not just a symbol) is bounded by $\theight{\re}+1$.
\begin{lemma}\label[lemma]{lemma:zero-inc}
 For all $\re,\re'\in\reSet$, and $\sym\in\symAlph$, if $\re\parDer{\sym}\re'$, then $\incFun{\re'}=0$.
\end{lemma}
\begin{proof}
 By induction and case analysis on the rules defining $\re\parDer{\sym}\re'$.
 \begin{description}
  \item[base case:] The only base rule is \rn{sym}.
   Hence, $\re=\sym$, $\re'=\eps$, and
   $\incFun{\eps}=0$.
  \item[inductive step:] the proof proceeds by case analysis on the rule applied
   at the root of the derivation tree.
   \begin{itemize}
    \item $\Rule{l-cat}{\re_0\parDer{\sym}\re'_0}{\re_0\catop\re_1\parDer{\sym}\re'_0\catop\re_1}{}$\\[2ex]
          If $\geqFun{\re_0}{\re_1}=1$, then by definition of $\incSym$ and inductive hypothesis, $\incFun{\re'_0\catop\re_1}=\incFun{\re'_0}=0$.
          If $\geqFun{\re_0}{\re_1}=0$, then by definition of $\incSym$ $\incFun{\re'_0\catop\re_1}=0$.

    \item $\Rule{r-cat}{\re_1\parDer{\sym}\re'_1}{\re_0\catop\re_1\parDer{\sym}\re'_1}{\hasEps{\re_0}=\eps}$\\[2ex]
          $\incFun{\re'_1}=0$ directly follows from the inductive hypothesis.

    \item the proof for rules \rn{l-or} and \rn{r-or} is the same as for rule \rn{r-cat}.

    \item the proof for rule \rn{star} is the same as for rule \rn{l-cat}.
   \end{itemize}
 \end{description}
\end{proof}

Given \cref{lemma:zero-inc}, the following corollary may look useless,
but it provides a general form of invariant which can be proved also for all other considered cases, where the strong claim of \cref{lemma:zero-inc} no longer holds.

\begin{corollary}\label[corollary]{cor:inv}
 For all $\re,\re'\in\reSet$, $\sym\in\symAlph$, if $\re\parDer{\sym}\re'$, then $\theight{\re'}\leq \theight{\re}+\incFun{\re}-\incFun{\re'}$.
\end{corollary}
\begin{proof}
 A direct consequence of \cref{theo:inc-bound} and \cref{lemma:zero-inc}.
\end{proof}

To provide intuition for the invariant of \cref{cor:inv}, let us consider the following equivalent form:
\begin{equation}\label{invariant}
 \theight{\re'}+\incFun{\re'}\leq\theight{\re}+\incFun{\re}
\end{equation}

As noted in the previous page, $\theight{\re'}+\incFun{\re'}$ and $\theight{\re}+\incFun{\re}$ define the upper bound of the height of all partial derivatives of $\re'$ and $\re$, respectively. But $\re'$ is a partial derivative of $\re$ because $\re\parDer{\sym}\re'$, hence, by transitivity, the set $P'$ of all partial derivatives of $\re'$ is a subset of all partial derivatives $P$ of $\re$.

The following theorem shows that the invariant property of \cref{cor:inv} can be extended to derivatives \wrt~ any words.
\begin{theorem}\label[theorem]{theo:gen-inc-bound}
 For all $\re,\re'\in\reSet$, and $\word\in\symAlph^*$, if $\re\parDer{\word}\re'$, then $\theight{\re'}\leq\theight{\re}+\incFun{\re}-\incFun{\re'}$.
\end{theorem}

\begin{proof}
 By induction on the length of $\word$.
 \begin{description}
  \item[base case:] if $\word=\emptyWord$, then by definition $\re'=\re$, hence
   $\theight{\re'}=\theight{\re}$ and $\incFun{\re}=\incFun{\re'}$, therefore $\theight{\re'}=\theight{\re}\leq\theight{\re}+\incFun{\re}-\incFun{\re'}$.

  \item[inductive step:]
   if $\word=\sym\strcons\word'$ for some $\sym\in\symAlph$, $\word'\in\symAlph^*$, then by definition there exists $\re''\in\reSet$ s.t. $\re\parDer{\sym}\re''$ and $\re''\parDer{\word'}\re'$. By \cref{cor:inv} $\theight{\re''}\leq\theight{\re}+\incFun{\re}-\incFun{\re''}$ and by inductive hypothesis $\theight{\re'}\leq\theight{\re''}+\incFun{\re''}-\incFun{\re'}$. Therefore by transitivity,
   $\theight{\re'}\leq\theight{\re''}+\incFun{\re''}-\incFun{\re'}\leq\theight{\re}+\incFun{\re}-\incFun{\re''}+\incFun{\re''}-\incFun{\re'}=\theight{\re}+\incFun{\re}-\incFun{\re'}$.
 \end{description}
\end{proof}

The proof of \cref{theo:gen-inc-bound} is independent from the definition of $\incSym$, providing that the invariant of \cref{cor:inv} holds. Therefore, the crucial point  is the definition of the function returning the upper bound of the increment of a partial derivation.

The result on the space complexity of derivatives can be directly derived from  \cref{theo:gen-inc-bound} and \cref{lemma:bounds}.

\begin{corollary}\label[corollary]{cor:gen-inc-bound}
 For all $\re,\re'\in\reSet$, and $\word\in\symAlph^*$, if $\re\parDer{\word}\re'$, then $\theight{\re'}\leq\theight{\re}+1$.
\end{corollary}

\begin{proof}
 By \cref{theo:gen-inc-bound} $\theight{\re'}\leq\theight{\re}+\incFun{\re}-\incFun{\re'}$. By \cref{lemma:bounds}
 $\incFun{\re}-\incFun{\re'}\leq 1$, hence
 $\theight{\re'}\leq\theight{\re}+1$.
\end{proof}

\subsection{Size of partial derivatives}\label{sec:size}
In this section we investigate the upper bound on the size of the partial derivatives of a given regular expression, defined in a standard way as follows:
$$
 \begin{array}{c}
  \tsize{\eps}=\tsize{\sym}=1                                                    \qquad
  \tsize{\re_0\catop\re_1}=\tsize{\re_0\orop\re_1}=\tsize{\re_0}+\tsize{\re_1}+1 \qquad
  \tsize{\starop{\re}}=\tsize{\re}+1
 \end{array}
$$

The results of~\cref{sec:height} must be adapted to the less trivial case of the size of expressions. Indeed, while for the height the upper bound of the increment is the constant 1, for the size this value depends on the square of the initial regular expression.

For instance, let us consider the following reduction steps which compute the partial derivative of $\starop{\starop{\starop{\sym}}}$
\wrt~ $\sym\sym$:
\begin{flushleft}
 $
  \begin{array}{l}
   \starop{\starop{\starop{\sym}}}\parDer{\sym}((\eps\catop\starop{\sym})\catop\starop{\starop{\sym}})\catop\starop{\starop{\starop{\sym}}}\parDer{\sym}((\eps\catop\starop{\sym})\catop\starop{\starop{\sym}})\catop\starop{\starop{\starop{\sym}}} \\
   \tsize{\starop{\starop{\starop{\sym}}}}=4\qquad  \tsize{((\eps\catop\starop{\sym})\catop\starop{\starop{\sym}})\catop\starop{\starop{\starop{\sym}}}}=13
  \end{array}
 $
\end{flushleft}
After the first step the size increases by 9 and remains unchanged after the second step.
For expressions of the general shape $\re=\starop{\sym}\ldots*$ it can be proved by induction on $n=\tsize{\re}\geq 2$ that if $\re\parDer{\sym}\re'$, then
$\tsize{\re'}=\tsize{\re}+\frac{n^2+n}{2}-1$.

As another example, let us consider the following reduction steps which compute the partial derivative of $(\starop{\sym}\catop\starop{\asym})\catop\starop{\yasym}$
\wrt~ $\sym\asym\yasym$:
\begin{flushleft}
 $
  \begin{array}{l}
   \sym\catop\starop{\starop{\asym}}\parDer{\sym}\eps\catop\starop{\starop{\asym}}\parDer{\asym}(\eps\catop\starop{\asym})\catop\starop{\starop{\asym}} \\
   \tsize{\sym\catop\starop{\starop{\asym}}}=5\qquad  \tsize{\eps\catop\starop{\starop{\asym}}}=5 \qquad \tsize{(\eps\catop\starop{\asym})\catop\starop{\starop{\asym}}}=8
  \end{array}
 $
\end{flushleft}
After the first step the size does not change, while increases by 3 after the second step. This example shows that the strong claim of \cref{lemma:zero-inc} cannot hold for $\tsize{\ }$.

% As another example, let us consider the following reduction steps which compute the partial derivative of $(\starop{\sym}\catop\starop{\asym})\catop\starop{\yasym}$
% \wrt~ $\sym\asym\yasym$:
% \begin{flushleft}
%  $
%   \begin{array}{l}
%    (\starop{\sym}\catop\starop{\asym})\catop\starop{\yasym}\parDer{\sym} ((\eps\catop\starop{\sym})\catop\starop{\asym})\catop\starop{\yasym}\parDer{\asym} (\eps\catop\starop{\asym})\catop\starop{\yasym}\parDer{\yasym}\eps\catop\starop{\yasym} \\
%    \tsize{(\starop{\sym}\catop\starop{\asym})\catop\starop{\yasym}}=8\qquad  \tsize{((\eps\catop\starop{\sym})\catop\starop{\asym})\catop\starop{\yasym}}=10 \qquad \tsize{(\eps\catop\starop{\asym})\catop\starop{\yasym}}=7 \qquad \tsize{\eps\catop\starop{\yasym}}=4
%   \end{array}
%  $
% \end{flushleft}
% After the first step the size increases by 2, while decreases by 3 after each of the other steps.

We follow the same methodology as in \cref{sec:height} and define the function $\incFunSize{\re}$, analogous to $\incFun{\re}$, for the size of $\re$.

\begin{definition}\label[definition]{fig:incFunSize}
 The function $\incFunSize{\re}$ is recursively defined as follows:
 $$
  \begin{array}{l}
   \incFunSize{\sym}=\incFunSize{\eps}=0                                                                       \\
   \incFunSize{\re_0\catop\re_1}=\max(\incFunSize{\re_0},\incFunSize{\re_1}-\tsize{\re_0}-1)                   \\
   \incFunSize{\re_0\orop\re_1}= \max(\incFunSize{\re_0}-\tsize{\re_1}-1,\incFunSize{\re_1}-\tsize{\re_0}-1,0) \\
   \incFunSize{\starop{\re}}=\tsize{\re}+\incFunSize{\re}+1
  \end{array}
 $$
\end{definition}
The following lemma is analogous to \cref{lemma:bounds}.
\begin{lemma}\label[lemma]{lemma:bounds-size}
 For all $\re\in\reSet$,
 $0\leq\incFunSize{\re}\leq \tsize{\re}^2$.
\end{lemma}
\begin{proof}
 Directly by induction on the definition of $\incFunSize{\re}$.
 We only show the proof for $\incFunSize{\re}\leq \tsize{\re}^2$ and omit the straightforward proof for $0\leq\incFunSize{\re}$.
 \begin{description}
  \item[base case:]  $\incFunSize{\sym}=\incFunSize{\eps}=0\leq 1 = \tsize{\sym}^2=\tsize{\eps}^2$
  \item[inductive step:]
   the proof proceeds by case analysis on the shape of $\re$.
   \begin{itemize}
    \item $\incFunSize{\re_0\catop\re_1}=\max(\incFunSize{\re_0},\incFunSize{\re_1}-\tsize{\re_0}-1)\leq\max(\incFunSize{\re_0},\incFunSize{\re_1})\leq\max(\tsize{\re_0}^2,\tsize{\re_1}^2)\leq \tsize{\re_0}^2+\tsize{\re_1}^2 \leq (\tsize{\re_0}+\tsize{\re_1})^2\leq (\tsize{\re_0}+\tsize{\re_1}+1)^2=\tsize{\re_0\catop\re_1}^2$, by the inductive hypothesis, the definition of $\max$, of square and of $\tsize{\ }$, and by the straightforward property $\tsize{\re}\geq 1$.
    \item $\incFunSize{\re_0\orop\re_1}\leq\tsize{\re_0\orop\re_1}^2$ can be proven similarly as in the previous case.
    \item $\incFunSize{\starop{\re}}=\tsize{\re}+\incFunSize{\re}+1\leq \tsize{\re}^2+\tsize{\re}+1\leq \tsize{\re}^2+2\cdot\tsize{\re}+1=(\tsize{\re}+1)^2=\tsize{\starop{\re}}^2$, by the inductive hypothesis, the definition of square and of $\tsize{\ }$, and by the straightforward property $\tsize{\re}\geq 1$.
   \end{itemize}
 \end{description}

\end{proof}

As for $\incFun{\re}$, the definition of $\incFunSize{\re}$ is driven by the reduction rules in \cref{fig:parDer}, and by the definition of size. We comment only the non-trivial cases: if $\re=\re_0\orop\re_1$, then the partial derivative can be obtained
by applying either rule \rn{l-or} or \rn{r-or}. For the former rule,
the partial derivative of $\re_0$ can increase of at most $\incFunSize{\re_0}$, but all the nodes of $\re_1$ and the $\orop$ operator are removed, which corresponds to $\incFunSize{\re_0}-\tsize{\re_1}-1$. The reasoning for the latter rule is symmetric. Then the maximum between these two values has to be considered to return a valid upper bound; furthermore, the bound must be non-negative for the same reason explained for $\incSym$.
If $\re=\starop{\re_0}$, then the new term $\re'$ needs to be considered, whose size is bounded by $\tsize{\re}+\incFunSize{\re}$. Furthermore, the concatenation operator is added, hence the function returns $\tsize{\re}+\incFunSize{\re}+1$. Finally, in the definition for $\re=\re_0\catop\re_1$ rule \rn{r-cat} is managed as rule \rn{r-or} (or \rn{l-or}).
In rule \rn{l-cat} $\re_0$ is replaced by $\re_0'$, while $\re_1$ and concatenation are kept, hence the upper bound of the increment is $\incFunSize{\re_0}$. As for $\re_0\orop\re_1$ the maximum needs to be considered.

Although the analogous of \cref{lemma:zero-inc} cannot be proved for $\incSymSize$, the invariant of \cref{cor:inv} holds, and can be proved directly.
\begin{theorem}\label[theorem]{theo:inc-bound-size}
 For all $\re,\re'\in\reSet$, $\sym\in\symAlph$, if $\re\parDer{\sym}\re'$, then $\tsize{\re'}\leq \tsize{\re}+\incFunSize{\re}-\incFunSize{\re'}$.
\end{theorem}
\begin{proof}
 By induction and case analysis on the rules defining $\re\parDer{\sym}\re'$.
 \begin{description}
  \item[base case:] The only base rule is \rn{sym}.
   Hence, $\re=\sym$, $\re'=\eps$, and
   $\tsize{\eps}=0=\tsize{\sym}=\tsize{\sym}+0-0=\tsize{\sym}+\incFunSize{\sym}-\incFunSize{\eps}$.
  \item[inductive step:] the proof proceeds by case analysis on the rule applied
   at the root of the derivation tree.
   \begin{itemize}
    \item $\Rule{l-cat}{\re_0\parDer{\sym}\re'_0}{\re_0\catop\re_1\parDer{\sym}\re'_0\catop\re_1}{}$\\[2ex]
          By inductive hypothesis $\tsize{\re_0'}\leq \tsize{\re_0}+\incFunSize{\re_0}-\incFunSize{\re_0'}$.
          We distinguish two cases:
          \begin{itemize}
           \item $\incFunSize{\re_0'}\geq\incFunSize{\re_1}-\tsize{\re_0'}-1$:
                 $\tsize{\re_0'\catop\re_1}\stackrel{\mathrm{def}}{=}\tsize{\re'_0}+\tsize{\re_1}+1\leq\tsize{\re_0}+\incFunSize{\re_0}-\incFunSize{\re_0'}+\tsize{\re_1}+1=\tsize{\re_0\catop\re_1}+\incFunSize{\re_0}-\incFunSize{\re_0'}=\tsize{\re_0\catop\re_1}+\incFunSize{\re_0}-\incFunSize{\re_0'\catop\re_1}\leq\tsize{\re_0\catop\re_1}+\incFunSize{\re_0\catop\re_1}-\incFunSize{\re_0'\catop\re_1}$, by the inductive hypothesis, the definition of $\max$, $\tsize{\ }$, $\incSymSize$, and the assumption $\incFunSize{\re_0'}\geq\incFunSize{\re_1}-\tsize{\re_0'}-1$.

           \item $\incFunSize{\re_0'}<\incFunSize{\re_1}-\tsize{\re_0'}-1$:
                 $\tsize{\re_0'\catop\re_1}\stackrel{\mathrm{def}}{=}\tsize{\re'_0}+\tsize{\re_1}+1=\tsize{\re'_0}+\tsize{\re_1}+1+\incFunSize{\re'_0\catop\re_1}-\incFunSize{\re'_0\catop\re_1}=\tsize{\re'_0}+\tsize{\re_1}+1+\incFunSize{\re_1}-\tsize{\re_0'}-1-\incFunSize{\re'_0\catop\re_1}=\tsize{\re_1}+\incFunSize{\re_1}-\incFunSize{\re'_0\catop\re_1}=\tsize{\re_0\catop\re_1}+\incFunSize{\re_1}-\tsize{\re_0}-1-\incFunSize{\re'_0\catop\re_1}\leq\tsize{\re_0\catop\re_1}+\incFunSize{\re_0\catop\re_1}-\incFunSize{\re'_0\catop\re_1}$, by the definition of $\max$, $\tsize{\ }$, $\incSymSize$, and the assumption $\incFunSize{\re_0'}<\incFunSize{\re_1}-\tsize{\re_0'}-1$.
          \end{itemize}

    \item $\Rule{r-cat}{\re_1\parDer{\sym}\re'_1}{\re_0\catop\re_1\parDer{\sym}\re'_1}{\hasEps{\re_0}=\eps}$\\[2ex]
          By inductive hypothesis $\tsize{\re_1'}\leq \tsize{\re_1}+\incFunSize{\re_1}-\incFunSize{\re'_1}$.

          Therefore, $\tsize{\re_1'}\leq \tsize{\re_1}+\incFunSize{\re_1}-\incFunSize{\re'_1}=\tsize{\re_0\catop\re_1}+\incFunSize{\re_1}-\tsize{\re_0}-1-\incFunSize{\re'_1}\leq\tsize{\re_0\catop\re_1}+\incFunSize{\re_0\catop\re_1}-\incFunSize{\re'_1}$, by the inductive hypothesis, the definition of $\max$, $\tsize{\ },$ and $\incSymSize$.

    \item the case for rules \rn{l-or} and \rn{r-or} is analogous to that for rule \rn{r-cat}.
    \item $\Rule{star}{\re\parDer{\sym}\re'}{\starop{\re}\parDer{\sym}\re'\catop\starop{\re}}{}$\\[2ex]
          By inductive hypothesis $\tsize{\re'}\leq \tsize{\re}+\incFunSize{\re}-\incFunSize{\re'}$.

          We distinguish two cases:
          \begin{itemize}
           \item $\incFunSize{\re'}\geq\incFunSize{\starop{\re}}-\tsize{\re'}-1$:
                 $\tsize{\re'\catop\starop{\re}}\stackrel{\mathrm{def}}{=}\tsize{\re'}+\tsize{\starop{\re}}+1\leq\tsize{\re}+\incFunSize{\re}-\incFunSize{\re'}+\tsize{\starop{\re}}+1=\tsize{\starop{\re}}+\incFunSize{\starop{\re}}-\incFunSize{\re'}=\tsize{\starop{\re}}+\incFunSize{\starop{\re}}-\incFunSize{\re'\catop\starop{\re}}$, by the inductive hypothesis, the definition of $\max$, $\tsize{\ }$, $\incSymSize$, and the assumption $\incFunSize{\re'}\geq\incFunSize{\starop{\re}}-\tsize{\re'}-1$.

           \item $\incFunSize{\re'}<\incFunSize{\starop{\re}}-\tsize{\re'}-1$:
                 $\tsize{\re'\catop\starop{\re}}\stackrel{\mathrm{def}}{=}\tsize{\re'}+\tsize{\starop{\re}}+1=\tsize{\re'}+\tsize{\starop{\re}}+1+\incFunSize{\re'\catop\starop{\re}}-\incFunSize{\re'\catop\starop{\re}}=\tsize{\re'}+\tsize{\starop{\re}}+1+\incFunSize{\starop{\re}}-\tsize{\re'}-1-\incFunSize{\re'\catop\starop{\re}}=\tsize{\starop{\re}}+\incFunSize{\starop{\re}}-\incFunSize{\re'\catop\starop{\re}}$, by the definition of $\max$, $\tsize{\ }$, $\incSymSize$, and the assumption $\incFunSize{\re'}<\incFunSize{\starop{\re}}-\tsize{\re'}-1$.
          \end{itemize}
   \end{itemize}
 \end{description}
\end{proof}

Having proved the invariant of \cref{theo:inc-bound-size}, we can derive the main result analogously as done in \cref{sec:height}.
\begin{theorem}\label[theorem]{theo:gen-inc-bound-size}
 For all $\re,\re'\in\reSet$, and $\word\in\symAlph^*$, if $\re\parDer{\word}\re'$, then $\tsize{\re'}\leq\tsize{\re}+\incFunSize{\re}-\incFunSize{\re'}$.
\end{theorem}
\begin{proof}
 See \cref{theo:gen-inc-bound}
\end{proof}

The result on the space complexity of derivatives directly follows from  \cref{theo:gen-inc-bound-size} and \cref{lemma:bounds-size}.

\begin{corollary}\label[corollary]{cor:gen-inc-bound-size}
 For all $\re,\re'\in\reSet$, and $\word\in\symAlph^*$, if $\re\parDer{\word}\re'$, then $\tsize{\re'}\leq\tsize{\re}+\tsize{\re}^2$. Hence $\tsize{\re'}$ is $O(\tsize{\re}^2)$.
\end{corollary}

\begin{proof}
 By \cref{theo:gen-inc-bound-size} $\tsize{\re'}\leq\tsize{\re}+\incFunSize{\re}-\incFunSize{\re'}$. By \cref{lemma:bounds}
 $\incFunSize{\re}-\incFunSize{\re'}\leq\incFunSize{\re}\leq\tsize{\re}^2$, hence
 $\tsize{\re'}\leq\tsize{\re}+\tsize{\re}^2$.
\end{proof}

\section{Regular expressions with shuffle}\label{sec:shuffle}

\subsection{Basic definitions}
In this section we extend the syntax of regular expressions by adding the shuffle operator\footnote{In the examples we assume that the shuffle has lower precedence than the concatenation and union operator.} $\re_0\shuffleop\re_1$ whose semantics is defined as follows by extending \cref{def:langOp} and \cref{def:sem}:
\[
 \begin{array}{l}
  \emptyWord \shuffleop \word = \word \shuffleop \emptyWord = \{\word\}                                                                                                      \qquad
  (\sym_1\word_1) \shuffleop (\sym_2\word_2) = \{ \sym_1  \} \strcatop (\word_1 \strshuffleop (\sym_2\word_2)) \cup \{\sym_2\}\strcatop((\sym_1\word_1)\strshuffleop\word_2) \\
  \lang_1\strshuffleop\lang_2=\bigcup_{\word_1\in\lang_1,\word_2\in\lang_2}(\word_1\strshuffleop\word_2)                                                                     \qquad
  \sem{\re_0\shuffleop\re_1}=\sem{\re_0}\strshuffleop\sem{\re_1}
 \end{array}
\]
We denote with $\reSetExt$ the set of regular expressions extended with the shuffle operator.

Although it is well-known that the shuffle operator does not increase the abstract expressive power of regular expressions, in practice it is still very useful in RV to write much more compact and clear
specifications when correct system behaviors can be described as independently interleaved event traces. Indeed, it has been proved \cite{BrodaEtAl18,MayerElAl94} that there exists a family of regular expressions $\re$ with shuffle for which the equivalent NFA needs at least $2^{\tsize{\re}}$ states.

Let us consider for instance the distinct events $o_n$, $a_n$, and $c_n$, with the meaning
``file $n$ has been opened'', ``accessed'', and ``closed'', respectively, where $n=1,2$ is the corresponding file descriptor.
If the SUS is allowed to manage the two files independently, then a specification for the correct use of them can be defined quite concisely by the regular expression $o_1\catop\starop{a_1}\catop c_1\shuffleop o_2\catop\starop{a_2}\catop c_2$.

If, for simplicity, we assume that files can be accessed only once, and we generalize over the number of files, then the specification becomes the regular expression $o_1\catop a_1\catop c_1\shuffleop\ldots\shuffleop o_n\catop a_n \catop c_n$ over the alphabet of $3n$ distinct symbols $\{o_1,a_1,c_1,\dots, o_n,a_n,c_n\}$. For such an expression, an equivalent NFA must have at least
$4^n$ states.

\begin{figure}[h]
 $$
  \begin{array}{c}
   \Rule{shf}{\re_0\der{\sym}\re'_0\quad\re_1\der{\sym}\re'_1}{\re_0\shuffleop\re_1\der{\sym}(\re'_0\shuffleop\re_1)\orop(\re_0\shuffleop\re'_1)}{}\qquad \hasEps{\re_0\shuffleop\re_1}=\hasEps{\re_0}\conj\hasEps{\re_1} \\[4ex]
   \Rule{l-shf}{\re_0\parDer{\sym}\re'_0}{\re_0\shuffleop\re_1\parDer{\sym}\re'_0\shuffleop\re_1}{} \qquad
   \Rule{r-shf}{\re_1\parDer{\sym}\re'_1}{\re_0\shuffleop\re_1\parDer{\sym}\re_0\shuffleop\re'_1}{}
  \end{array}
 $$
 \caption{Transition rules defining the derivative and the partial derivatives for the shuffle operator}
 \label{fig:shfParDer}
\end{figure}

The transition rules defining the derivative and the partial derivatives for the shuffle operator
can be found in \cref{fig:shfParDer}. They all correspond to the intuition that events can be interleaved; moreover, the empty trace is contained in $\re_0\shuffleop\re_1$ iff it is contained in $\re_0$ and $\re_1$.

As happens for the other operators, also with the shuffle  the derivative is always defined, while this is not the case for partial derivatives. For instance, if we assume $\sym_2\neq\sym_0$ and $\sym_2\neq\sym_1$, then
$\sym_0\shuffleop \sym_1\der{\sym_2}(\none\shuffleop\sym_1)\orop(\sym_0\shuffleop\none)$, and $\hasEps{(\none\shuffleop\sym_1)\orop(\sym_0\shuffleop\none)}=\none$, but there exists no $\re$ s.t.
$\sym_0\shuffleop \sym_1\parDer{\sym_2}\re$.

Theorems~\ref{theo:der} and \ref{theo:parDer} still hold along with \cref{cor:parDer}, when $\reSetExt$ is considered.

\subsection{Height of partial derivatives with the shuffle operator}

We extend the result of \cref{theo:gen-inc-bound} and \cref{cor:gen-inc-bound} to the case of the shuffle operators.

Unfortunately the proofs are more challenging because the claim of \cref{lemma:zero-inc} no longer holds: There exist partial derivatives $\re$ \st~$\incFun{\re}>0$.
As a counter example, let us consider the following reduction steps computing the partial derivatives of $(\eps\shuffleop\starop{\sym})\catop(\asym\shuffleop\starop{\sym})$ \wrt~ $\sym\asym\sym$:
\begin{flushleft}
 $
  \begin{array}{l}
   \re_0\parDer{\sym}\re_1\parDer{\asym}\re_2\parDer{\sym}\re_3 \\
   \re_0=(\eps\shuffleop\starop{\sym})\catop(\asym\shuffleop\starop{\sym}) \qquad
   \re_1=(\eps\shuffleop\eps\catop\starop{\sym})\catop(\asym\shuffleop\starop{\sym})                                     \qquad
   \re_2=\eps\shuffleop\starop{\sym}                                       \qquad \re_3=\eps\shuffleop\eps\catop\starop{\sym}
  \end{array}
 $
\end{flushleft}
We have $\theight{\re_1}=\theight{\re_0}+1$ (first reduction step), but also $\theight{\re_3}=\theight{\re_2}+1$ (third reduction step). Therefore, necessarily $\incFun{\re_2}>0$, to ensure that \cref{invariant} holds.

To prove the extended versions of the results of \cref{sec:parDer} we first extend the definition of $\incSym$ given in \cref{fig:incFun} with the case for the shuffle.
\[
 \incFun{\re_0\shuffleop\re_1}=\max(\geqFun{\re_0}{\re_1}\cdot\incFun{\re_0},\geqFun{\re_1}{\re_0}\cdot\incFun{\re_1})
\]
Before explaining the definition of $\incSym$ above, we note that the claim of \cref{lemma:bounds} holds also for $\reSetExt$ and can still be proved by induction on the definition of $\incSym$: $0\leq\incFun{\re}\leq 1$.
Consequently, if the height of one sub-expression $\re_i$ is strictly greater than the other $\re_{1-i}$, then only the partial derivative of $\re_i$ can contribute to the increment of the height of the partial derivative of $\re_0\shuffleop\re_1$, because increments are always bounded by 1.
Note that $\max(\incFun{\re_i},0)=\incFun{\re_i}$, since $\incFun{\re_i}\geq 0$, therefore $\incFun{\re_0\shuffleop\re_1}=\incFun{\re_i}$, if $\theight{\re_i}>\theight{\re_{1-i}}$ for $i=0,1$.

If $\theight{\re_0}=\theight{\re_1}$, then both the derivatives of $\re_0$ and $\re_1$ can contribute to the increment of the height of the partial derivative of $\re_0\shuffleop\re_1$.
Since \rn{l-shf} or \rn{r-shf} can be non-deterministically applied, the maximum increment
$\max(\incFun{\re_0},\incFun{\re_1})$ needs to be considered.

The correctness of the definition of $\incFun{\re_0\shuffleop\re_1}$ is still ensured by the claim of \cref{theo:inc-bound} extended to $\reSetExt$.

\subsubsection*{Proof of \cref{theo:inc-bound} extended to $\reSetExt$
}
For all $\re,\re'\in\reSetExt$, $\sym\in\symAlph$, if $\re\parDer{\sym}\re'$, then $\theight{\re'}\leq \theight{\re}+\incFun{\re}$.

\begin{proof}
 The structure of the proof is the same, but the two new rules for the shuffle operator need to be considered.
 \begin{itemize}
  \item %$\Rule{l-shf}{\re_0\parDer{\sym}\re'_0}{\re_0\shuffleop\re_1\parDer{\sym}\re'_0\shuffleop\re_1}{}$\\[2ex]
        rule \rn{l-shf}: By inductive hypothesis $\theight{\re_0'}\leq \theight{\re_0}+\incFun{\re_0}$.
        We distinguish two cases:
        \begin{itemize}
         \item $\theight{\re_0}\geq\theight{\re_1}$:
               $\theight{\re_0'\shuffleop\re_1}\stackrel{\mathrm{def}}{=}\max(\theight{\re'_0},\theight{\re_1})+1\leq\max(\theight{\re_0}+\incFun{\re_0},\theight{\re_1})+1\leq\max(\theight{\re_0},\theight{\re_1})+1+\incFun{\re_0}\leq\max(\theight{\re_0},\theight{\re_1})+1+\incFun{\re_0\shuffleop\re_1}=\theight{\re_0\shuffleop\re_1}+\incFun{\re_0\shuffleop\re_1}$, where inequalities are derived from the inductive hypothesis, the definition of $\max$, $\theight{\ }$, $\incSym$, and $\geqSym$, the assumption $\theight{\re_0}\geq\theight{\re_1}$, and \cref{lemma:bounds}.

         \item $\theight{\re_0}<\theight{\re_1}$ (that is, $\theight{\re_0}+\incFun{\re_0}\leq\theight{\re_1}$ by \cref{lemma:bounds}):
               $\theight{\re_0'\shuffleop\re_1}\stackrel{\mathrm{def}}{=}\max(\theight{\re'_0},\theight{\re_1})+1\leq\max(\theight{\re_0}+\incFun{\re_0},\theight{\re_1})+1=\theight{\re_0\shuffleop\re_1}\leq\theight{\re_0\shuffleop\re_1}+\incFun{\re_0\shuffleop\re_1}$, where inequalities are derived from the inductive hypothesis, the definition of $\max$, $\theight{\ }$, $\incSym$, and $\geqSym$, the assumption $\theight{\re_0}<\theight{\re_1}$, and \cref{lemma:bounds}.
        \end{itemize}
  \item rule (r-shf) is symmetric to rule (l-shf).
 \end{itemize}
\end{proof}

Since both \cref{lemma:bounds} and \cref{theo:inc-bound} hold for $\reSetExt$, \cref{cor:bound} holds for $\reSetExt$ as well.

To prove for $\reSetExt$ the invariant defined by~\cref{invariant}, we modularize the proof by introducing two lemmas.

The first lemma is a weaker version of \cref{lemma:zero-inc}.

\begin{lemma}\label[lemma]{lemma:ext-zero-inc}
 For all $\re,\re'\in\reSetExt$, and $\sym\in\symAlph$, if $\re\parDer{\sym}\re'$ and $\theight{\re'}=\theight{\re}+1$, then $\incFun{\re'}=0$.
\end{lemma}

\begin{proof}
 By induction and case analysis on the rules defining $\re\parDer{\sym}\re'$.
 \begin{description}
  \item[base case:] The only base rule is \rn{sym}.
   The case is vacuous because $\re=\sym$, $\re'=\eps$, and $\theight{\sym}=\theight{\eps}=0$.
  \item[inductive step:] \hspace*{\fill}
   \begin{itemize}
    \item $\Rule{l-cat}{\re_0\parDer{\sym}\re'_0}{\re_0\catop\re_1\parDer{\sym}\re'_0\catop\re_1}{}$\\[2ex]
          From the hypothesis and the definition of $\theight{\ }$
          \begin{equation}
           \label{eq:l-cat-one}
           \max(\theight{\re_0'},\theight{\re_1})=\max(\theight{\re_0},\theight{\re_1})+1
          \end{equation}
          Therefore $\theight{\re_1}\leq\max(\theight{\re_0},\theight{\re_1})<\max(\theight{\re_0},\theight{\re_1})+1=\max(\theight{\re_0'},\theight{\re_1})$ by the definition of $\max$ and \cref{eq:l-cat-one}.

          Therefore, by the definition of $\max$
          \begin{equation}
           \label{eq:l-cat-two}
           \max(\theight{\re_0'},\theight{\re_1})=\theight{\re_0'}
          \end{equation}
          Hence $\theight{\re_0}+1\leq\max(\theight{\re_0},\theight{\re_1})+1=\theight{\re'_0}$ by the definition of $\max$, \cref{eq:l-cat-one} and \cref{eq:l-cat-two}.

          Moreover, by \cref{cor:bound} $\theight{\re_0'}\leq\theight{\re_0}+1$, therefore
          $\theight{\re_0'}=\theight{\re_0}+1$, and by inductive hypothesis $\incFun{\re'_0}=0$, which implies $\incFun{\re'_0\catop\re_1}=0$ by definition of $\incSym$.

    \item $\Rule{r-cat}{\re_1\parDer{\sym}\re'_1}{\re_0\catop\re_1\parDer{\sym}\re'_1}{\hasEps{\re_0}=\eps}$\\[2ex]
          This case is vacuous because $\theight{\re_1'}\leq\theight{\re_1}+1\leq\max(\theight{\re_0},\theight{\re_1})+1=\theight{\re_0\catop\re_1}$ by \cref{cor:bound} and the definition of $\max$ and $\theight{\ }$.

    \item The case for rules \rn{l-or} and \rn{r-or} is vacuous for the same reason shown for \rn{r-cat}.
    \item $\Rule{l-shf}{\re_0\parDer{\sym}\re'_0}{\re_0\shuffleop\re_1\parDer{\sym}\re'_0\shuffleop\re_1}{}$\\[2ex]
          The same proof for \rn{l-cat} shows that $\incFun{\re_0'}=\incFun{\re_0}+1$ and $\theight{\re_1}<\theight{\re'_0}$, therefore $\incFun{\re'_0}=0$ by inductive hypothesis, and $\incFun{\re'_0\shuffleop\re_1}=\max(1\cdot\incFun{\re_0'},0\cdot\incFun{\re_1})=0$ by $\theight{\re_1}<\theight{\re'_0}$ and the definition of $\max$, $\incSym$ and $\geqSym$.
    \item rule \rn{r-shf} is symmetric to rule \rn{l-shf}.
    \item $\Rule{star}{\re\parDer{\sym}\re'}{\starop{\re}\parDer{\sym}\re'\catop\starop{\re}}{}$\\[2ex]
          By definition of $\theight{\ }$ and by the hypothesis $\theight{\re'\catop\starop{\re}}=\theight{\starop{\re}}+1$ we have
          $\max(\theight{\re'},\theight{\starop{\re}})+1=\theight{\starop{\re}}+1$, therefore
          $\max(\theight{\re'},\theight{\starop{\re}})=\theight{\starop{\re}}$.

          If $\theight{\re'}<\theight{\starop{\re}}$ then by the definition of $\incSym$ and $\geqSym$ we have $\incFun{\re'\catop\starop{\re}}=\geqFun{\re'}{\starop{\re}}\cdot\incFun{\re'}=0\cdot\incFun{\re'}=0$.

          If $\theight{\re'}\geq\theight{\starop{\re}}$ then $\theight{\re'}=\theight{\starop{\re}}$, since $\max(\theight{\re'},\theight{\starop{\re}})=\theight{\starop{\re}}$.
          By the definition of $\theight{\ }$ we have $\theight{\re'}=\theight{\re}+1$, therefore
          by inductive hypothesis $\incFun{\re'}=0$ and by definition of $\incSym$ and $\geqSym$ we have $\incFun{\re'\catop\starop{\re}}=\geqFun{\re'}{\starop{\re}}\cdot\incFun{\re'}=1\cdot\incFun{\re'}=0$.
   \end{itemize}
 \end{description}
\end{proof}

The second lemma establishes an invariant on $\incSym$ when the height of the derivatives does not change.

\begin{lemma}\label[lemma]{lemma:ext-leq-inc}
 For all $\re,\re'\in\reSetExt$, and $\sym\in\symAlph$, if $\re\parDer{\sym}\re'$ and $\theight{\re'}=\theight{\re}$, then $\incFun{\re'}\leq\incFun{\re}$.
\end{lemma}

\begin{proof}
 By induction and case analysis on the rules defining $\re\parDer{\sym}\re'$.
 \begin{description}
  \item[base case:] The only base rule is \rn{sym}.
   In this case $\re=\sym$, $\re'=\eps$, therefore $\incFun{\eps}=0=\incFun{\sym}$.
  \item[inductive step:] \hspace*{\fill}
   \begin{itemize}
    \item $\Rule{l-cat}{\re_0\parDer{\sym}\re'_0}{\re_0\catop\re_1\parDer{\sym}\re'_0\catop\re_1}{}$\\[2ex]
          From the hypothesis and the definition of $\theight{\ }$
          \begin{equation}
           \label{eq:l-cat-three}
           \max(\theight{\re_0'},\theight{\re_1})=\max(\theight{\re_0},\theight{\re_1})
          \end{equation}
          Two different cases may occur:
          \begin{itemize}
           \item $\theight{\re_0}\geq\theight{\re_1}$\\
                 From \cref{eq:l-cat-three} and the definition of $\max$, $\max(\theight{\re_0'},\theight{\re_1})=\theight{\re_0}$, hence $\theight{\re_0'}\leq\theight{\re_0}$ by the definition of $\max$.

                 If $\theight{\re_0'}<\theight{\re_0}$, then $\theight{\re_0'}<\theight{\re_1}$ by \cref{eq:l-cat-three} and the definition of $\max$. Therefore $\incFun{\re_0'\catop\re_1}=0\leq\incFun{\re_0\catop\re_1}$ by definition of $\incSym$ and \cref{lemma:bounds}.

                 Therefore $\theight{\re_1}\leq\max(\theight{\re_0},\theight{\re_1})<\max(\theight{\re_0},\theight{\re_1})+1=\max(\theight{\re_0'},\theight{\re_1})$ by the definition of $\max$ and \cref{eq:l-cat-one}.

                 If $\theight{\re_0'}=\theight{\re_0}$, then $\incFun{\re_0'}\leq\incFun{\re_0}$ by inductive hypothesis. Since $\theight{\re_0'}=\theight{\re_0}\geq\theight{\re_1}$, we have $\incFun{\re_0'\catop\re_1}=\incFun{\re_0'}\leq\incFun{\re_0}=\incFun{\re_0\catop\re_1}$ by the definition of $\incSym$ and $\geqSym$.
           \item $\theight{\re_0}<\theight{\re_1}$\\
                 In this case $\incFun{\re_0\catop\re_1}=0$ by the definition of $\incSym$ and $\geqSym$. Moreover, from \cref{eq:l-cat-three} $\max(\theight{\re_0'},\theight{\re_1})=\theight{\re_1}$, hence $\theight{\re_0'}\leq\theight{\re_1}$ by the definition of $\max$.

                 If $\theight{\re_0'}<\theight{\re_1}$ then $\incFun{\re_0'\catop\re_1}=0=\incFun{\re_0\catop\re_1}$ by the definition of $\incSym$ and $\geqSym$. If $\theight{\re_0'}=\theight{\re_1}$ then $\theight{\re_0}<\theight{\re_1}=\theight{\re_0'}$, therefore $\theight{\re_0}+1\leq\theight{\re_0'}$. Moreover, by \cref{cor:bound} $\theight{\re'_0}\leq\theight{\re_0}+1$, hence $\theight{\re_0'}=\theight{\re_0}+1$, and, by \cref{lemma:ext-zero-inc}, $\incFun{\re_0'}=0$. Finally,
                 $\incFun{\re_0'\catop\re_1}=\incFun{\re_0'}=0=\incFun{\re_0\catop\re_1}$ by the definition of $\incSym$ and $\geqSym$.
          \end{itemize}

    \item $\Rule{r-cat}{\re_1\parDer{\sym}\re'_1}{\re_0\catop\re_1\parDer{\sym}\re'_1}{\hasEps{\re_0}=\eps}$\\[2ex]
          From the hypothesis and the definition of $\theight{\ }$
          \begin{equation}
           \label{eq:l-cat}
           \theight{\re_1'}=\max(\theight{\re_0},\theight{\re_1})+1
          \end{equation}
          Two different cases may occur:
          \begin{itemize}
           \item $\theight{\re_0}\geq\theight{\re_1}$\\
                 From \cref{eq:l-cat} and the definition of $\max$, $\theight{\re_1'}=\max(\theight{\re_0},\theight{\re_1})+1\geq\theight{\re_1}+1$. Moreover, by \cref{lemma:bounds},
                 $\theight{\re_1'}\leq\theight{\re}+1$, therefore $\theight{\re_1'}=\theight{\re}+1$, hence by \cref{lemma:ext-zero-inc} $\incFun{\re_1'}=0$.
                 Finally, $\incFun{\re_1'}=0\leq\incFun{\re_0\catop\re_1}$ by \cref{lemma:bounds}.

           \item $\theight{\re_0}<\theight{\re_1}$\\
                 In this case $\incFun{\re_0\catop\re_1}=0$ by the definition of $\incSym$ and $\geqSym$. Moreover, from \cref{eq:l-cat} and the definition of $\max$, $\theight{\re_1'}=\theight{\re_1}+1$, hence  by \cref{lemma:ext-zero-inc}, $\incFun{\re_1'}=0$. Finally,
                 $\incFun{\re_1'}=0\leq\incFun{\re_0\catop\re_1}$ by \cref{lemma:bounds}.
          \end{itemize}

    \item The proofs for the rules \rn{l-or} and \rn{r-or} are analogous to the proof for the rule \rn{r-cat}.
    \item $\Rule{l-shf}{\re_0\parDer{\sym}\re'_0}{\re_0\shuffleop\re_1\parDer{\sym}\re'_0\shuffleop\re_1}{}$\\[2ex]
          The proof is similar to that for \rn{l-cat}, except from some details on the definition of $\incFun{\re_0\shuffleop\re_1}$, which differs from that of $\incFun{\re_0\catop\re_1}$. We report it for the sake of completeness.

          From the hypothesis and the definition of $\theight{\ }$
          \begin{equation}
           \label{eq:l-shuffle}
           \max(\theight{\re_0'},\theight{\re_1})=\max(\theight{\re_0},\theight{\re_1})
          \end{equation}
          Two different cases may occur:
          \begin{itemize}
           \item $\theight{\re_0}\geq\theight{\re_1}$\\
                 From \cref{eq:l-shuffle} and the definition of $\max$, $\max(\theight{\re_0'},\theight{\re_1})=\theight{\re_0}$, hence $\theight{\re_0'}\leq\theight{\re_0}$ by the definition of $\max$.

                 If $\theight{\re_0'}<\theight{\re_0}$, then $\theight{\re_0'}<\theight{\re_1}=\theight{\re_0}$ by \cref{eq:l-shuffle} and the definition of $\max$. Therefore $\geqFun{\re_0'}{\re_1}= 0$ and $\geqFun{\re_1}{\re_0'}=\geqFun{\re_0}{\re_1}=\geqFun{\re_1}{\re_0}=1$ by the definition of $\geqSym$, and
                 $\incFun{\re_0'\shuffleop\re_1}=\incFun{\re_1}\leq\max(\incFun{\re_0},\incFun{\re_1})=\incFun{\re_0\shuffleop\re_1}$ by the definition of $\incSym$ and $\max$.

                 If $\theight{\re_0'}=\theight{\re_0}$, then $\incFun{\re_0'}\leq\incFun{\re_0}$ by inductive hypothesis. Moreover,
                 $\geqFun{\re_0}{\re_1}=\geqFun{\re'_0}{\re_1}=1$ and $\geqFun{\re_1}{\re_0}=\geqFun{\re_1}{\re'_0}$ by the definition of $\geqSym$.
                 Therefore, $\incFun{\re_0'\shuffleop\re_1}=\max(\incFun{\re'_0},\geqFun{\re_1}{\re_0'}\cdot\incFun{\re_1})\leq\max(\incFun{\re_0},\geqFun{\re_1}{\re_0}\cdot\incFun{\re_1})=\incFun{\re_0\shuffleop\re_1}$ by the definition of $\max$.
           \item $\theight{\re_0}<\theight{\re_1}$\\
                 In this case $\incFun{\re_0\shuffleop\re_1}=\incFun{\re_1}$ by the definition of $\incSym$ and $\geqSym$. Moreover, from \cref{eq:l-shuffle} $\max(\theight{\re_0'},\theight{\re_1})=\theight{\re_1}$, hence $\theight{\re_0'}\leq\theight{\re_1}$ by the definition of $\max$.

                 If $\theight{\re_0'}<\theight{\re_1}$ then $\incFun{\re_0'\shuffleop\re_1}=\incFun{\re_1}=\incFun{\re_0\shuffleop\re_1}$ by the definition of $\incSym$ and $\geqSym$. If $\theight{\re_0'}=\theight{\re_1}$ then $\theight{\re_0}<\theight{\re_1}=\theight{\re_0'}$, therefore $\theight{\re_0}+1\leq\theight{\re_0'}$. Moreover, by \cref{cor:bound} $\theight{\re'_0}\leq\theight{\re_0}+1$, hence $\theight{\re_0'}=\theight{\re_0}+1$, and, by \cref{lemma:ext-zero-inc}, $\incFun{\re_0'}=0$. Finally,
                 $\incFun{\re_0'\shuffleop\re_1}=\max(\incFun{\re_0'},\incFun{\re_1})=\max(0,\incFun{\re_1})=\incFun{\re_0\shuffleop\re_1}$ by the definition of $\incSym$ and $\geqSym$.
          \end{itemize}

    \item The rule \rn{r-shf} is symmetric to \rn{l-shf}.
    \item $\Rule{star}{\re\parDer{\sym}\re'}{\starop{\re}\parDer{\sym}\re'\catop\starop{\re}}{}$\\[2ex]
          Directly by the definition of $\incSym$ and \cref{lemma:bounds}
          $\incFun{\re'\catop\starop{\re}}\leq 1=\incFun{\starop{\re}}$.
   \end{itemize}
 \end{description}
\end{proof}

\begin{theorem}\label[theorem]{theo:inv}
 For all $\re,\re'\in\reSetExt$, $\sym\in\symAlph$, if $\re\parDer{\sym}\re'$, then $\theight{\re'}\leq\theight{\re}+\incFun{\re}-\incFun{\re'}$.
\end{theorem}
\begin{proof}
 By \cref{cor:bound} (extended to $\reSetExt$) $\theight{\re'}\leq\theight{\re}+ 1$.
 Three cases are distinguished.
 \begin{itemize}
  \item $\theight{\re'}=\theight{\re}+1$: by \cref{lemma:ext-zero-inc}
        $\incFun{\re'}=0$, therefore by \cref{theo:inc-bound} (extended to $\reSetExt$) $\theight{\re'}\leq\theight{\re}+\incFun{\re}=\theight{\re}+\incFun{\re}-\incFun{\re'}$.
  \item $\theight{\re'}=\theight{\re}$: by \cref{lemma:ext-leq-inc} $\incFun{\re}-\incFun{\re'}\geq 0$, hence
        $\theight{\re'}=\theight{\re}\leq\theight{\re}+\incFun{\re}-\incFun{\re'}$.
  \item $\theight{\re'}<\theight{\re}$: by \cref{lemma:bounds} (extended to $\reSetExt$) $-1\leq\incFun{\re}-\incFun{\re'}$, hence
        $\theight{\re'}\leq\theight{\re}-1\leq\theight{\re}+\incFun{\re}-\incFun{\re'}$.
 \end{itemize}
\end{proof}

From \cref{theo:inv} and \cref{lemma:bounds} (extended to $\reSetExt$) we can
extend for free \cref{theo:gen-inc-bound} and \cref{cor:gen-inc-bound}  to $\reSetExt$.

\begin{theorem}
 For all $\re,\re'\in\reSetExt$, and $\word\in\symAlph^*$, if $\re\parDer{\word}\re'$, then $\theight{\re'}\leq\theight{\re}+\incFun{\re}-\incFun{\re'}$.
\end{theorem}

\begin{corollary}
 For all $\re,\re'\in\reSetExt$, and $\word\in\symAlph^*$, if $\re\parDer{\word}\re'$, then $\theight{\re'}\leq\theight{\re}+1$.
\end{corollary}

\subsection{Size of partial derivatives with the shuffle operator}

Similar as done in the previous section, we extend the result of \cref{theo:gen-inc-bound-size} and \cref{cor:gen-inc-bound-size} to the case of the shuffle operator in the case of the size $\tsize{\re}$ of a regular expression $\re$.

Also for the size, the introduction of shuffle does not affect the space complexity of partial derivatives. However, in this case the proofs can be extended in an easier way, although the definition of $\incFunSize{\re_0\shuffleop\re_1}$ requires some care.

Indeed, one would be tempted to consider the following (incorrect) definition:
\[
 \incFunSize{\re_0\shuffleop\re_1}=\max(\incFunSize{\re_0},\incFunSize{\re_1})
 \qquad\qquad\qquad\qquad\mbox{(incorrect definition)}
\]
Unfortunately, this definition verifies only the weaker inequality
$\tsize{\re}\leq\tsize{\re_0\shuffleop\re_1}+\incFunSize{\re_0\shuffleop\re_1}$, if $\re_0\shuffleop\re_1\parDer{\sym}\re$, but not $\tsize{\re}\leq\tsize{\re_0\shuffleop\re_1}+\incFunSize{\re_0\shuffleop\re_1}-\incFunSize{\re}$.

To show this, let us consider the following reduction step computing the partial derivative of $\starop{\sym}\shuffleop\starop{\asym}$ \wrt~ $\sym$:
\begin{flushleft}
 $
  \begin{array}{l}
   \starop{\sym}\shuffleop\starop{\asym}\parDer{\sym}\eps\catop\starop{\sym}\shuffleop\starop{\asym} \\
   \tsize{\starop{\sym}\shuffleop\starop{\asym}}=5 \qquad
   \incFunSize{\starop{\sym}\shuffleop\starop{\asym}}=2 \qquad
   \tsize{\eps\catop\starop{\sym}\shuffleop\starop{\asym}}=7 \qquad
   \incFunSize{\eps\catop\starop{\sym}\shuffleop\starop{\asym}}=2
  \end{array}
 $
\end{flushleft}
We have $\tsize{\eps\catop\starop{\sym}\shuffleop\starop{\asym}}=7\leq 5+2=\tsize{\starop{\sym}\shuffleop\starop{\asym}}+\incFunSize{\starop{\sym}\shuffleop\starop{\asym}}$, but
$\tsize{\eps\catop\starop{\sym}\shuffleop\starop{\asym}}=7\not\leq 5+2-2=\tsize{\starop{\sym}\shuffleop\starop{\asym}}+\incFunSize{\starop{\sym}\shuffleop\starop{\asym}}-\incFunSize{\eps\catop\starop{\sym}\shuffleop\starop{\asym}}$.

To guarantee the invariant $\tsize{\re'}\leq\tsize{\re}+\incFunSize{\re}-\incFunSize{\re'}$, the following definition has to be considered:
\[
 \incFunSize{\re_0\shuffleop\re_1}=\incFunSize{\re_0}+\incFunSize{\re_1} \qquad\qquad\qquad\qquad\mbox{(correct definition)}
\]
This new definition takes into account increments due to multiple reduction steps. In this way, we obtain $\tsize{\eps\catop\starop{\sym}\shuffleop\starop{\asym}}=7\leq 5+4-2=\tsize{\starop{\sym}\shuffleop\starop{\asym}}+\incFunSize{\starop{\sym}\shuffleop\starop{\asym}}-\incFunSize{\eps\catop\starop{\sym}\shuffleop\starop{\asym}}$.

The correctness of the definition of $\incFunSize{\re_0\shuffleop\re_1}$ is ensured by the claim of \cref{theo:inc-bound-size} extended to $\reSetExt$.

\subsubsection*{Proof of \cref{theo:inc-bound-size} extended to $\reSetExt$
}
For all $\re,\re'\in\reSetExt$, $\sym\in\symAlph$, if $\re\parDer{\sym}\re'$, then $\tsize{\re'}\leq \theight{\re}+\incFunSize{\re}-\incFunSize{\re'}$.

\begin{proof}
 The structure of the proof is the same, but the two new rules for the shuffle operator need to be considered.
 \begin{itemize}
  \item %%$\Rule{l-shf}{\re_0\parDer{\sym}\re'_0}{\re_0\shuffleop\re_1\parDer{\sym}\re'_0\shuffleop\re_1}{}$\\[2ex]
        rule~\rn{l-shf}: By inductive hypothesis $\tsize{\re_0'}\leq \tsize{\re_0}+\incFunSize{\re_0}-\incFunSize{\re_0'}$.

        We have $\tsize{\re'_0\shuffleop\re_1}=\tsize{\re_0'}+\tsize{\re_1}+1\leq\tsize{\re_0}+\incFunSize{\re_0}-\incFunSize{\re_0'}+\tsize{\re_1}+1=\tsize{\re_0\shuffleop\re_1}+\incFunSize{\re_0\shuffleop\re_1}-\incFunSize{\re_1}-\incFunSize{\re_0'\shuffleop\re_1}+\incFunSize{\re_1}=\tsize{\re_0\shuffleop\re_1}+\incFunSize{\re_0\shuffleop\re_1}-\incFunSize{\re_0'\shuffleop\re_1}$, by the inductive hypothesis, the definition of $\tsize{\ }$, and $\incSymSize$.
  \item rule (r-shf) is symmetric to rule (l-shf).
 \end{itemize}
\end{proof}

From \cref{theo:inc-bound-size} and \cref{lemma:bounds-size} (extended to $\reSetExt$) we can
extend for free \cref{theo:gen-inc-bound-size} and \cref{cor:gen-inc-bound-size}  to $\reSetExt$.

\begin{theorem}
 For all $\re,\re'\in\reSetExt$, and $\word\in\symAlph^*$, if $\re\parDer{\word}\re'$, then $\tsize{\re'}\leq\tsize{\re}+\incFunSize{\re}-\incFunSize{\re'}$.
\end{theorem}

\begin{corollary}
 For all $\re,\re'\in\reSetExt$, and $\word\in\symAlph^*$, if $\re\parDer{\word}\re'$, then $\tsize{\re'}\leq\tsize{\re}+\tsize{\re}^2$. Hence $\tsize{\re'}$ is $O(\tsize{\re}^2)$.
\end{corollary}
\section{Conclusion}\label{sec:conclu}
We have explored the space complexity of partial derivatives of regular expressions in the context of RV with regular expressions extended with the shuffle operator. While the size of the set of syntactically distinct partial derivatives of a regular expression $\re$ is known to be linear in the size
of $\re$, no bounds on the size of the largest derivative had been previously investigated. We fill this gap by analyzing the height and size of partial derivatives and establishing upper bounds for both metrics.

We have shown that the height of any partial derivative of a regular expression increases by at most one, and that this property still holds when the shuffle operator is considered. Furthermore, we proved that the size of the largest partial derivative is bounded quadratically in the size of the original expression. This quadratic bound also holds in the presence of shuffle, despite with this operator there exist regular expressions whose equivalent NFAs exhibit an exponential explosion of the number of states.

Our approach is based on a general proof methodology that defines functions to compute upper bounds on the increase in height and size of the partial derivatives. These functions allows us to establish an invariant that can be generalized to multiple rewriting steps, enabling a modular and reusable proof structure. This methodology allowed us to seamlessly extend our results to a more complex operator like shuffle.

Our results support the practical use of partial derivatives in rewriting-based RV, ensuring that the memory usage remains manageable even in the worst case. Moreover, they show that the shuffle operator is useful for writing short and readable specifications without incurring high computational costs.

Short-term future work includes the extension of our analysis to
other improvements of the expressive power of regular expressions, for instance with the use of additional operators or parameterized specifications.

\bibliography{main}
\end{document}